\documentclass[11pt]{amsart}
\usepackage{fullpage}
\usepackage{amsthm,amsmath,amssymb}
\usepackage[foot]{amsaddr}
\usepackage{enumitem}
\usepackage{booktabs}
\usepackage{graphicx}
\usepackage{subcaption}
\usepackage{tikz}
\usepackage{natbib}
\usetikzlibrary{arrows.meta, positioning}

\usepackage[colorlinks=true,linkcolor=red,breaklinks=true,citecolor=blue,urlcolor=cyan]{hyperref}

\graphicspath{{./figs/}} 

\newtheorem{theorem}{Theorem}[section]
\newtheorem{proposition}[theorem]{Proposition}

\newtheorem{corollary}[theorem]{Corollary}

\newtheorem{remark}[theorem]{Remark}

\title{Multiple Regression Analysis of Unmeasured Confounding}

\author{Brian Knaeble}
\address{Department of Computer Science, Utah Valley University, Orem, UT}
\email{bknaeble@uvu.edu}
\author{R. Mitchell Hughes}
\keywords{partial identification, omitted variables bias, uncertain models, coefficients of partial determination, sensitivity analysis}

\begin{document}

\begin{abstract}
Whereas confidence intervals are used to assess uncertainty due to unmeasured individuals, confounding intervals can be used to assess uncertainty due to unmeasured attributes. Previously, we have introduced a methodology for computing confounding intervals in a simple regression setting in a paper titled ``Regression Analysis of Unmeasured Confounding." Here we extend that methodology for more general application in the context of multiple regression. Our multiple regression analysis of unmeasured confounding utilizes subject matter knowledge about coefficients of determination to bound omitted variables bias, while taking into account measured covariate data. Our generalized methodology can be used to partially identify causal effects. The methodology is demonstrated with example applications, to show how coefficients of determination, being complementary to randomness, can support sensitivity analysis for causal inference from observational data. The methodology is best applied when natural sources of randomness are present and identifiable within the data generating process. Our main contribution is an algorithm that supports our methodology. The purpose of this article is to describe our algorithm in detail. In the paper we provide a link to our GitHub page for readers who would like to access and utilize our algorithm.
\end{abstract}

\maketitle

\section{Introduction}
\label{introduction}
A technique for computing ``confounding intervals'' was published in a prior article on regression analysis of unmeasured confounding \citep{Knaeble2020}. A confounding interval assesses uncertainty due to unmeasured attributes, and it can be contrasted with common confidence intervals that assess uncertainty due to unmeasured individuals. Separately, the concept of partial identification has been developed by economists \citep{Tamer2010}, with utility for causal inference within the Rubin causal model \citep{Imbens2015}. In the discussion subsection of \cite{Knaeble2020} the possibility of more general confounding intervals is mentioned but not developed in detail. The purpose of this article is to introduce a more general methodology to compute confounding intervals from multiple regression models. The more general methodology can be thought of as a way to analyze omitted variables bias in a multiple regression setting, and when additional causal assumptions are met then the more general methodology can be viewed as a method for partial identification of a causal effect. To the extent that regression is useful for causal inference \cite{McNamee2005,Angrist2009}, the methodology of this paper does facilitate causal inference from observational data and natural experiments. However, the methodology is best appreciated as a technical methodology to summarize the sensitivity of a multiple regression coefficient to omitted variables bias, and in that sense we speak of partial identification of a regression coefficient that has been defined from unmeasured covariate data.

It is instructive for us to clarify what exactly is meant by partial identification and how it can support causal interpretation of a regression coefficient. To clarify the matter we briefly revisit an example from an earlier paper on reversals of least-squares estimates \citep{KD}. The raw data from the example came from a study that investigated the effect of rice consumption on urinary arsenic concentrations in children \citep{Davis2012}. The researchers regressed the log-transformed arsenic measurements onto dietary rice consumption and reported an expected increase in urinary arsenic of $15.6\%$ for each 1/4 cup increase in dietary rice consumption. They then adjusted for measured covariates including age, sex, and gender by using multiple regression \citep{McNamee2005}, and their estimate changed from a $15.6\%$ increase to a $14.3\%$ increase. Adjustment for additional covariates could perhaps change their estimate by an arbitrarily large amount. With such model uncertainty \citep{Chat} can we even be sure that the effect is positive (see \citep{KD})? The authors conclude by stating that their study only ``suggests that rice consumption is a potential source of arsenic exposure,'' displaying an awareness of model uncertainty and the limitations of classical statistical methodology for the purpose of causal inference from observational data. On the one hand there is a chosen, best, statistically significant estimate, and on the other hand there is a theoretical, causal quantity, namely the causal effect of a 1/4 cup increase in rice consumption on blood arsenic concentrations. But, it is not possible to identify the causal effect with an estimate, even in the limit of an infinitely large sample size, because rice consumption was not randomized within an experiment. Adjusting for measured covariates is not the same as adjusting for measured and unmeasured covariates. To address the issue of unmeasured covariates, we may, however, partially identify the causal effect. Partial identification, in contrast to more precise point identification, bounds the hypothetical causal effect within upper and lower limits. The purpose of this paper is to demonstrate how to compute upper and lower limits for that partial identification from a multiple regression model and some additional knowledge. The additional knowledge is in the form of upper bounds on coefficients of determination for the exposure (e.g. rice consumption) and the outcome (e.g. urinary arsenic concentration).

Upper bounds on coefficients of determination are reasonable and consistent with what has been revealed by modern experiments regarding (the inadequacy of) theories of hidden variables to explain quantum phenomena \citep{Gill2014}. One could argue that coefficients of determination approach values of $100\%$ in the limit as more and more information is obtained, but the philosophy of this paper is contrary to such (super) determinism. The methodology of this paper is best understood from a philosophical viewpoint that accepts inherent limits on our ability to predict events. Opposing determinism is randomness, and it has been said that randomization is the reasoned basis for causal inference in experiments \citep{Fisher,Rose10}. There are natural sources of randomness \citep{Rosen}, and it is possible to quantify sufficient randomness for causal inference from observational studies \citep{Knaeble2023,Knaeble2024,10.48550/arxiv.2407.19057}. Ultimately, it is the presence of natural randomness that allows us to partially identify a causal effect from observational data using the methodology of this paper. That randomness is reflected in the upper bounds that we place on coefficients of determination. This way of thinking is valuable when analyzing data from so-called natural, natural-experiments \citep{Rosen}, especially when there is a random component in the natural data generating process. 

The data generating process can be made explicit with structural causal models, and causal graphs have long been used to facilitate causal inference \citep{Wright,Spirtes}. Judea Pearl has received the Turing Award for his work on causation \cite{Pearl2000}, and an important result regarding covariate selection is described in \citep[Section 3.3.1]{Pearl2009}. It is not the details of the structural causal model that matter but the structure of the graph, and there is an algorithm that takes a graph as an input and produces an admissible set of covariates as an output. Here ``admissible'' means that it is sufficient to condition on that admissible set of covariates to produce a consistent estimator of the causal effect of interest. If there exists a true graph of the data generating process and we can measure the treatment, the outcome, and that admissible set of covariates, then we can achieve point identification (not just partial identification) of the causal effect of interest, provided we have enough observations, i.e. a large enough sample of data. However, Imbens and Rubin have criticized the use of causal graphs \citep{Imbens1995}, and Rosenbaum describes a causal graph as an enormous web of assumptions \citep{Rose95}. While there are algorithms to learn graphs from data \citep{Nog}, any causal graph that purports to represent the data generating process is susceptible to a critique that it is not exactly correct. We take such a critique seriously, and accept that no causal graph will be perfect, and that acceptance is consistent with our plan for partial identification rather than point identification. 

We take a practical approach to covariate selection that is similar to the approach taken to formulate the disjunctive cause criterion for covariate selection \citep{VanderWeele2011}. Our introduced methodology does not assume that a causal graph is exactly correct, nor does it assume that an admissible set of covariates has been measured. We do, of course, recommend that covariates be selected in a principled manner \citep{VandPrinc}, but we acknowledge that the selection process will not be perfect. For the purpose of covariate selection there is the common-cause criterion which is a conservative approach, the pre-treatment criterion which is a more liberal approach, and the disjunctive cause criterion which strikes an intermediate balance \citep{VandPrinc}. In the context of a natural, natural experiment \citep{Rosen} there will be some randomness in the data generating process, which may mean that the more liberal pre-treatment criterion will suffice \citep{DM} for our purposes, because the fine-tuned counterexamples of \citep{PearlDM} require sufficient determinism in the data generating process. The ample randomness criterion of \citep[Section 4.3]{Knaeble2024} is a good guide in our context. It states a preference for covariates that measure fixed characteristics of individuals like time and place of birth or genomic variables, rather than randomly fluctuating covariates, e.g. a transient food craving. Similar ideas are expressed in \citep{Pim} and the classic funnel experiment of Deming \citep{Deming}. Regardless of how covariates have been selected, our methodology does not assume that any particular set of measured covariates is admissible in the sense of \citep{Pearl2009}. We merely assume the existence of a set of unmeasured covariates so that the union of the selected and measured covariates with that unmeasured set of covariates is admissible. That said, our methodology is more precise when covariates are selected following the principles of \citep{VandPrinc} and \citep[Section 4.3]{Knaeble2024}.
%
%
%
%

Knowledge of causal graphs is not required to understand our methodology. Our methodology is not rooted in any particular philosophy of causal inference. Those readers who work with propensity probabilities to approximate a randomized experiment will find value in our work; see \citep{Knaeble2023, Knaeble2024}. Also, those readers accustomed to causal graphs will find value in our methodology for sensitivity analysis in that it relaxes some assumptions. Our approach to sensitivity analysis has been inspired in part by the methods described in \citep{Frank}, \citep{SAWA}, and \citep[Chapter 3.4]{Rose10}. Our work is best understood as an attempt to expand the scope and applicability of sensitivity analysis for causal inference from observational data. Our intuitive sensitivity parameters open the door to partial identification in addition to sensitivity analysis, and we view each of our example applications from multiple perspectives to emphasize how our methodology supports both sensitivity analysis and partial identification. Our main contribution is our algorithm, which summarizes a multiple regression model to support causal interpretation of its slope coefficients. The ideas behind the algorithm may be useful beyond regression \citep{Pearl2013}.

The remainder of this paper is organized as follows. Technical notation, assumptions, mathematical results, and the algorithm itself are introduced in Section \ref{methods}. Example applications are described in Section \ref{application}. Connections to natural experiments and instruments, and the relevance of exclusion restrictions and SUTVA are discussed in Section \ref{discussion}. Also, related methods for analysis of omitted variables bias in a regression setting are discussed in Section \ref{discussion}, and we conclude the paper with an explanation for how our introduced methodology differs from those related methods, along with some tips for practical application of our introduced methodology. 
\section{Methods}
\label{methods}
Here we describe how to transform bounds on coefficients of determination so that the algorithm of \citep{KOA} can be applied. We introduce notation and derive some mathematical results. The context is that of an observational study. 
\subsection{Notation}We assume throughout that the sample size $n$ is sufficiently large. There is a single, continuous response variable $y$, a single, explanatory treatment variable $x$, $p$ additional, measured covariates $w=\{w_1,...,w_p\}$, and $q$ unmeasured, omitted variables $u=\{u_1,...,u_q\}$. We write $y$, $x$, $w$, $w_i$, $u$, and $u_j$ for both the variable names and also the associated column vectors and matrices of data. We allow the matrix $w=[w_1,...,w_p]$ to contain column vectors that represent interactions between measured covariates, and we allow the matrix $u=[u_1,...,u_q]$ to contain column vectors that represent interactions between covariates in the union of measured and unmeasured covariates, but those interactions can not involve $x$ nor $y$. To make sure the covariate indices match the vector indices we allow the measured covariate set $w=\{w_1,...,w_p\}$ and the unmeasured covariate set $u=\{u_1,...,u_q\}$ to contain named covariates that correspond to the interaction terms. We assume throughout that any set of column vectors is a linearly independent set. Since we are interested in slope coefficients of regression models we may assume without loss of generality that the column vectors $y$ and $x$ and each column vector within the matrices $w=[w_1,...,w_p]$ and $u=[u_1,...,u_q]$ are mean zero vectors. The variables of $w$ are measured, confounding variables, and the set of unmeasured covariates $u$ is defined as a set of unmeasured covariates that renders $\{w,u\}$ admissible in the sense of \citep{Pearl2009}. We discuss the times of definition for those variables in Section \ref{discussion}. We assume that such a $u$ exists, and use it to define the following, theoretical regression model:
\begin{equation}
\label{themod}
y=\beta_{x.y|w,u}x+\beta_{w_1.y|w,u}w_1+...+\beta_{w_p.y|w,u}w_p+\beta_{u_1.y|w,u}u_1+...+\beta_{u_q.y|w,u}u_q.
\end{equation}  
That model and all of our regression models are fit with the principle of least squares. The theoretical model in (\ref{themod}) is not fit to measured data because $u$ is unmeasured, but we can fit the following, practical, regression model, since $w$ is measured:
\begin{equation}
\label{pramod}
y=\beta_{x.y|w}x+\beta_{w_1.y|w}w_1+...+\beta_{w_p.y|w}w_p.
\end{equation} 
In both of those models and in all subsequent models a slope coefficient has subscripts that are to be interpreted as follows: left of the vertical bar and left of the dot is a subscript to clarify which explanatory variable is relevant to that particular coefficient, left of the vertical bar and right of the dot is a subscript to clarify which response vector is of interest, and right of the vertical bar are subscripts to clarify which covariate vectors have been included in the model. In (\ref{pramod}), the confounding variables in $w$ are to be selected in a principled way \citep{VandPrinc} to support causal interpretation of $\beta_{x.y|w}$, but point identification of our main coefficient of interest $\beta_{x.y|w,u}$ with $\beta_{x.y|w}$ is not justified because additional confounding variables may be in the omitted set $u$. Note that the practical model in (\ref{pramod}) need not be well fit (regression diagnostics are not required). All that is required is for the model in (\ref{themod}) to be well defined, meaning, in theory, that if $u$ were measured, then the model of (\ref{themod}) would be well fit.

Let $\hat{y}_w$, $\hat{x}_w$, and (for each $i=1,...,q$) $\hat{u_i}_w$ be the least-squares projections of $y$, $x$, and $u_i$ onto the span of $w$. The vectors of residuals are then $y-\hat{y}_w$, $x-\hat{x}_w$, and (for each $i=1,...,q$) $u_i-\hat{u_i}_w$, and when needed to simplify notation we denote those vectors of residuals with $r(y;w)$, $r(x;w)$, and (for each $i=1,...,q$) $r(u_i;w)$, respectively. Define the set of residual vectors $u-\hat{u}_w=\{u_i-\hat{u_i}_w\}_{i=1}^q$, and when needed denote it with $r(u;w)$. Let $\hat{y}_{w,u}$ and $\hat{x}_{w,u}$ be the least-squares projections of $y$ and $x$ onto the span of $[w,u]$. Write $|\cdot|$ for the L2 norm, and define the coefficients of determination $R^2_{w,u;y}=|\hat{y}_{w,u}|/|y|$ and $R^2_{w,u;x}=|\hat{y}_{w,u}|/|x|$. Subscripts of $R^2$ coefficients are to be interpreted as follows: the subscript to the right of the semi colon represents a vector that will be projected using the principle of least squares onto the span of the vectors represented by the subscript to the left of the semi colon. 

\subsection{Bounding coefficients of determination}
As mentioned in Section \ref{introduction}, the ideas of \citep{Gill2014} along with variations on the techniques of \citep{Knaeble2023,Knaeble2024,10.48550/arxiv.2407.19057} applied in the context of natural, natural experiments (see \cite{Rosen}) can bound coefficients of determination below unity. We do not go into those details here, but assume for the purpose of partial identification that subject matter specialists are able to produce upper bounds $B_y$ and $B_x$ to satisfy
\begin{equation}
\label{upperbounds}
0\leq R^2_{w;y}\leq R^2_{w,u;y}\leq B_y<1 \textrm{~and~} 0\leq R^2_{w;y}\leq R^2_{w,u;y}\leq B_x<1.
\end{equation}
We may, for the purpose of sensitivity analysis, treat $B_y$ and $B_x$ as variables, and we demonstrate how to do that in Section \ref{application}.

\subsection{Adjusting for covariates}
There is more than one way to adjust for covariates when conducting causal inference with regression models. Define vectors of residuals $y-\hat{y}_{w,u}$ and $x-\hat{x}_{w,u}$ and denote them with $r(y;w,u)$ and $r(x;w,u)$ respectively. To adjust for $\{w,u\}$ we may regress $y$ onto $[x,w,u]$ or we may regress $r(y;w,u)$ onto $r(x;w,u)$. There is also a middle way: we may ``partial out'' \citep{CohenApplied} only $w$ and regress $r(y,w)$ onto $[r(x,w),r(u,w)]$. That middle way produces a slope coefficient of interest, expressed in our notation as  $\beta_{r(x;w).r(y;w)|r(u;w)}$.
\begin{proposition}
\label{Danny}
With $\beta_{r(x;w).r(y;w)|r(u;w)}$ as just defined and $\beta_{x.y|w,u}$ as defined in (\ref{themod}) the following identity holds:
\begin{equation}\label{slopeidentify}\beta_{r(x;w).r(y;w)|r(u;w)}=\beta_{x.y|w,u}.\end{equation}
\end{proposition}
\noindent A proof of Proposition \ref{Danny} is described in \citep[Proposition 3.1]{KBB}. By (\ref{slopeidentify}) we may partially identify $\beta_{x.y|w,u}$ by partially identifying $\beta_{r(x;w).r(y;w)|r(u;w)}$. We do so with the aid of (\ref{upperbounds}) and the main algorithm of \citep{KOA}. 

\subsection{The CI algorithm}
\label{CI}
Here we summarize the main algorithm of \citep{KOA} in the language and notation of this paper. We refer to the main algorithm of \citep{KOA} as the confounding interval (CI) algorithm. The CI algorithm can be used to produce lower and upper bounds $L(B_x,B_y)$ and $U(B_x,B_y)$ to partially identify $\beta_{r(x;w).r(y;w)|r(u;w)}$, i.e. to satisfy the following inequalities:
\begin{equation}\label{LU}L(B_x,B_y)\leq \beta_{r(x;w).r(y;w)|r(u;w)}\leq U(B_x,B_y).\end{equation}
Those lower and upper bounds depend on which covariates are included in $\{w,u\}$, and more explicitly on $B_x$ and $B_y$ of (\ref{upperbounds}). The CI algorithm takes inputs obtainable from measured $[y,x,w]$ via the residual vectors $r(y;w)$ and $r(x;w)$. Those residual vectors are obtained by regressing $y$ onto $w$ and then $x$ onto $w$. From those residual vectors we obtain the following sufficient parameters, which are inputs to the CI algorithm: the ratio of standard deviations $\sigma_{r(y;w)}/\sigma_{r(x;w)}$ and Pearson's linear correlation coefficient $\rho(r(x;w),r(y;w))$. The CI algorithm has additional inputs that are specified on the basis of subject matter knowledge: lower and upper bounds on $R^2_{r(u;w);r(x;w)}$, lower and upper bounds on $R^2_{r(u;w);r(y;w)}$, and lower and upper bounds on the more complicated Pearson's linear correlation coefficient $\rho(f(r(x;w),r(u,w)),f(r(y;w),r(u,w))$, where $f(r(x;w),r(u,w))$ is the vector of fitted values when $r(x,w)$ is projected onto $r(u,w)$ using the principle of least-squares, and $f(r(y;w),r(u,w))$ is the vector of fitted values when $r(y,w)$ is projected onto $r(u,w)$ using the principle of least-squares. Here we simply use the default bounds of $-1\leq\rho(f(r(x;w),r(u,w)),f(r(y;w),r(u,w))\leq 1$ that hold because $\rho$ is Pearson's linear correlation coefficient.

\subsection{Bounding residual coefficients of determination}
To adapt the CI algorithm to our context, it remains to derive bounds on $R^2_{r(u;w);r(x;w)}$ and  $R^2_{r(u;w);r(y;w)}$ from the bounds of (\ref{upperbounds}). 

\begin{proposition}
\label{mainprop}
With $R^2_{r(u,w);r(y;w)}$, $R^2_{r(u,w);r(x;w)}$, $R^2_{w,u;y}$, $R^2_{w,u;x}$, $R^2_{w;y}$, and $R^2_{w;x}$ as defined previously, the following identities hold: 
\begin{equation}\label{mainpropeq}R^2_{r(u,w);r(y;w)}=\frac{R^2_{w,u;y}-R^2_{w;y}}{1-R^2_{w;y}}\textrm{~and~}R^2_{r(u,w);r(x;w)}=\frac{R^2_{w,u;x}-R^2_{w;x}}{1-R^2_{w;x}}.\end{equation}
\end{proposition}
\begin{proof}
We prove $R^2_{r(u,w);r(y;w)}=\frac{R^2_{w,u;y}-R^2_{w;y}}{1-R^2_{w;y}}$. The proof of $R^2_{r(u,w);r(x;w)}=\frac{R^2_{w,u;x}-R^2_{w;x}}{1-R^2_{w;x}}$ is the same but with $x$ in place of $y$. Recall that $r(y;w)=y-\hat{y}_w$; by sum of squares we have
\begin{equation}
\label{firstsumsquares}
||y||^2=||y-\hat{y}_w||^2+||\hat{y}_w||^2.
\end{equation}
Also, since each column of $w$ is orthogonal to each column of $r(u;w)=u-\hat{u}_w$ we have
\begin{equation}
\label{secondsumsquares}
||\hat{y}_{w,r(u;w)}||^2=||\hat{y}_w||^2+||\hat{y}_{r(u;w)}||^2.
\end{equation}
After rearranging (\ref{firstsumsquares}) and (\ref{secondsumsquares}) we divide to obtain
\begin{equation}
\label{ratio}
\frac{||\hat{y}_{r(u,w)}||^2}{||y-\hat{y}_w||^2}=\frac{||\hat{y}_{w,r(u,w)}||^2-||\hat{y}_w||^2}{||y||^2-||\hat{y}_w||^2}.
\end{equation}
Again by the orthogonality of $w$ and $r(u;w)$ we have $\hat{y}_{r(u;w)}=\hat{r(y;w)}_{r(u;w)}$ so the left hand side of (\ref{ratio}) can be rewritten as
\begin{equation}
\label{rephrased}
\frac{||\hat{y}_{r(u;w)}||^2}{||y-\hat{y}_w||^2}=\frac{||\hat{r(y;w)}_{r(u,w)}||^2}{||y-\hat{y}_w||^2}=R^2_{r(u;w),r(y;w)}.
\end{equation} 
On the other hand, since the span of $\{w,r(u;w)\}$ is equal to the span of $\{w,u\}$ we have $\hat{y}_{w,r(u;w)}=\hat{y}_{w,u}$ and the right hand side of (\ref{ratio}) can be rewritten as
\begin{equation}
\label{conclusion}
\frac{||\hat{y}_{w,r(u;w)}||^2-||\hat{y}_w||^2}{||y||^2-||\hat{y}_w||^2}=\frac{||\hat{y}_{w,u}||^2-||\hat{y}_w||^2}{||y||^2-||\hat{y}_w||^2}=\frac{(||\hat{y}_{w,u}||^2-||\hat{y}_w||^2)/||y||^2}{(||y||^2-||\hat{y}_w||^2)/||y||^2}=\frac{R^2_{w,u;y}-R^2_{w;y}}{1-R^2_{w;y}}.
\end{equation}
\end{proof}
\begin{remark}
    \label{partialR2}
    The quantities in (\ref{mainpropeq}) are referred to as coefficients of partial determination.
\end{remark}

Proposition \ref{mainprop} expresses $R^2_{r(u,w);r(y;w)}$ and $R^2_{r(u,w);r(x;w)}$ as linear functions of $R^2_{w,u;y}$ and $R^2_{w,u;x}$ given $R^2_{w;y}$ and $R^2_{w;x}$. The quantities $R^2_{w;y}$ and $R^2_{w;x}$ can be computed from the measured data $[y,x,w]$. Subject matter knowledge is used to bound the values of $R^2_{w,u;y}$ and $R^2_{w,u;x}$; see (\ref{upperbounds}). Due to the linearity of the relationships in (\ref{mainpropeq}), the bounds of (\ref{upperbounds}) can be inserted into the formulas of (\ref{mainpropeq}) to produce desired bounds on $R^2_{r(u,w);r(y;w)}$ and $R^2_{r(u,w);r(x;w)}$ as described in the following corollary.  
\begin{corollary}
\label{cor}
With $B_y$ and $B_x$ as defined in (\ref{upperbounds}), and $R^2_{r(u,w);r(y;w)}$, $R^2_{r(u,w);r(x;w)}$, $R^2_{w;y}$, and $R^2_{w;x}$ as defined previously, we bound $R^2_{r(u,w);r(y;w)}$ and $R^2_{r(u,w);r(x;w)}$ as follows:
\begin{equation}
\label{algobounds}
0\leq R^2_{r(u,w);r(y;w)} \leq \frac{B_y-R^2_{w;y}}{1-R^2_{w;y}}<1\textrm{~and~}0\leq R^2_{r(u,w);r(x;w)} \leq \frac{B_x-R^2_{w;x}}{1-R^2_{w;x}}<1.
\end{equation}
\end{corollary}

\subsection{Summary}
\label{summary}
We now briefly summarize the introduced methodology for partial identification. From an observational study we measure $[y,x,w]$ and from that data we compute the coefficients of determination $R^2_{w;y}$ and $R^2_{w;x}$, and also the residual vectors $r(y;w)$ and $r(x;w)$. From those residual vectors we obtain the ratio of standard deviations $\sigma_{r(y;w)}/\sigma_{r(x;w)}$ and the correlation coefficient $\rho(r(x;w),r(y;w))$. We seek to partially identify the slope coefficient $\beta_{x.y|w,u}$ in the theoretical model of (\ref{themod}), and to do so we bound the theoretical coefficients of determination $R^2_{w,u;y}$ and $R^2_{w,u;x}$ with upper bounds $B_y$ and $B_x$ as described in (\ref{upperbounds}). We then apply the CI algorithm of Subsection \ref{CI} with $\sigma_{r(y;w)}/\sigma_{r(x;w)}$ and $\rho(r(x;w),r(y;w))$ as inputs along with the transformed bounds of Corollary \ref{cor}. The output of the CI algorithm consists of lower and upper bounds $L(B_x,B_y)$ and $U(B_x,B_y)$ on $\beta_{r(x;w).r(y;w)|r(u;w)}$ which by Proposition \ref{Danny} may be used to bound $\beta_{x.y|w,u}$. 

Overall, the methodology is an algorithm to assess the uncertainty in causal interpretation of the regression coefficient $\beta_{x.y|w}$ in the multiple regression model described in (\ref{pramod}) due to an unmeasured confounding set $u$. The overall algorithm takes measured $[y,x,w]$ data and user specified upper bounds $\{B_y,B_x\}$ satisfying $R^2_{w,u;y}\leq B_y$ and $R^2_{w,u;x}\leq B_x$ (see \ref{upperbounds}) as inputs, and then produces the interval $[L(B_x,B_y),U(B_x,B_y)]$ as an output, to partially identify the slope coefficient $\beta_{x.y|w,u}$ of the theoretical regression model (\ref{themod}) as follows: \begin{equation}\label{partialinterval}L(B_x,B_y)\leq \beta_{x.y|w,u}\leq U(B_x,B_y).\end{equation}
In short, the introduced methodology is an algorithm that carries out the following map:
\[(y,x,w;B_x,B_y)\mapsto [L(B_x,B_y),U(B_x,B_y)].\]
Our notation emphasizes the dependency of $L$ and $U$ on $(B_x,B_y)$ in case there is uncertainty regarding the user specified bounds $B_x$ and $B_y$. Our algorithm is available at \cite{Hughes25}, with extensions that treat $B_x$ and $B_y$ as variables as demonstrated in the next section.
\section{Example Applications}
\label{application}
Two example applications of our methodology are described here. The first example application addresses the question of whether smoking during pregnancy causes low birth weight. When describing the first example we emphasize how the concept of statistical significance may by itself be of limited value in support of causal inference, and we describe the first example from a point of view that should be familiar to analysts who work with causal graphs and admissible sets. The second example demonstrates how our introduced methodology supports causal inference from natural, natural experiments \citep{Rosen}. The weather is naturally random and partly unpredictable, and the second example application addresses the question of whether wind causes lower levels of air pollution. R programs that support the following analyses are available at \cite{Hughes25}.
\subsection{Smoking and birth weight}
\label{smoking}
Does smoking during pregnancy cause low birth weight?
We address that question with data from the (US) National Center for Health Statistics \cite{NCHSBirth}. The treatment variable $x$ is ``smoking'' measured as the average number of cigarettes smoked per day (cig/day) over the course of a pregnancy. We exclude those individuals who smoke more than twenty cigarettes per day, and we exclude those individuals who do not smoke at all. That exclusion does not substantially alter the fitted regression coefficients, but it does lead to better-fit regression models. The exclusion leaves us with a sample of moderate smokers that is still quite large, of size $n=298,334$. The outcome variable $y$ is the birth weight of each newborn baby. The birth weight is measured in grams, g, and as a variable it is denoted as ``birthweight''. The measured covariates in $w$ are the mother's age, father's age, mother's education level, and child's sex, along with racial and ethnic indicator variables for both the mother and father.

We regress birth weight onto smoking and the covariates of $w$ to obtain
$\hat{\beta}_{\textrm{smoking.birthweight}|w}=-6.09$ g/(cig/day). Due to the very large sample size the standard error of that fitted slope coefficient is $0.17$ g/(cig/day). Classical statistical analysis leads to a rejection of the null hypothesis that $\beta_{\textrm{smoking.birthweight}|w}=0$ with a very small p-value ($p<10^{-15}$). But that inference is not entirely informative nor directly relevant to our causal question of whether an intervention to lower smoking would cause an increase in birth weight. It may be convenient to assume that $w$ is an admissible set of covariates, but we make no such non-confounding assumption here. We can not infer from $\hat{\beta}_{\textrm{smoking.birthweight}|w}$ and its relatively small standard error that each additional cigarette smoked per day causes an expected loss of about $6$ grams of newborn weight. To more clearly address the causal question we ignore sampling error and focus instead on error due to unmeasured attributes. We have $\beta_{\textrm{smoking.birthweight}|w}=-6.09$ g/(cig/day) and by Proposition \ref{Danny} we have also that  $\beta_{r(\textrm{smoking};w).r(\textrm{birthweight};w)}=-6.09$ g/(cig/day). That adjusted slope is visible in Figure \ref{residuals1}. Conditional on $w$ there is a negative relationship between smoking and birth weight.
\begin{figure}[ht]
    \centering
    \includegraphics[width=.6\linewidth]{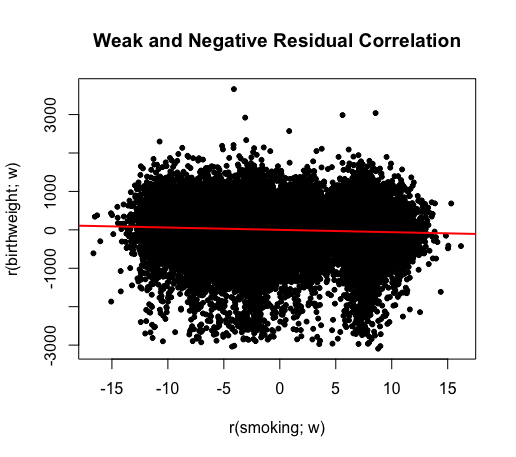}
    \caption{A scatter plot of $r(\textrm{smoking};w)$ vs $r(\textrm{birthweight};w)$ along with the least-squares regression line with slope $\beta_{\textrm{smoking.birthweight}|w}=\beta_{r(\textrm{smoking};w).r(\textrm{birthweight};w)}=-6.09$ g/(cig/day).}
    \label{residuals1}
\end{figure}

The broad question is whether $\beta_{r(\textrm{smoking};w).r(\textrm{birthweight};w)}=-6.09$ g/(cig/day) admits a causal interpretation. A narrow, more technical question, is whether any set of unmeasured covariates $u$ that makes $\{w,u\}$ admissible in the sense of \citep{Pearl2009} could also make $\beta_{\textrm{smoking.birthweight}|w,u}$ far from $\beta_{\textrm{smoking.birthweight}|w}=\beta_{r(\textrm{smoking};w).r(\textrm{birthweight};w)}=-6.09$ g/(cig/day). To answer that narrow, technical question, we consider variable upper bounds $B_{\textrm{smoking}}$ and $B_{\textrm{birthweight}}$ that satisfy $R^2_{w,u;\textrm{smoking}}\leq B_{\textrm{smoking}}$ and $R^2_{w,u;\textrm{birthweight}}\leq B_{\textrm{birthweight}}$. From the residuals pictured in Figure \ref{residuals1} we compute the ratio of standard deviations $\sigma_{r(\textrm{birthweight};w)}/\sigma_{r(\textrm{smoking};w)}= 92.75$ and the correlation coefficient \[\rho(r(\textrm{smoking};w),r(\textrm{birthweight};w))=-0.07\] and then fix those values throughout the remainder of the analysis in this Subsection \ref{smoking}. We have computed $R^2_{w;\textrm{smoking}}=0.05$ and $R^2_{w;\textrm{birthweight}}=0.03$ from the measured data. For each feasible (see (\ref{upperbounds})) pair $(B_{\textrm{smoking}},B_{\textrm{birthweight}})$ satisfying \[R^2_{w;\textrm{smoking}}=0.05\leq B_{\textrm{smoking}}<1\textrm{~and~}R^2_{w;\textrm{birthweight}}=0.03\leq B_{\textrm{birthweight}}<1\] we utilize the overall algorithm of Subsection \ref{summary} to compute \[L(B_{\textrm{smoking}}, B_\textrm{birthweight})\textrm{~and~}U(B_{\textrm{smoking}}, B_\textrm{birthweight})\] of (\ref{partialinterval}) and plot the results in Figure \ref{3dplot1}. The contours of $U(B_{\textrm{smoking}}, B_\textrm{birthweight})$ are shown in Figure \ref{contour1}. Figures \ref{3dplot1} and \ref{contour1} do not contain strong evidence in support of the claim that smoking during pregnancy causes low birth weight.

\begin{figure}[ht]
    \centering
    \includegraphics[width=.65\linewidth]{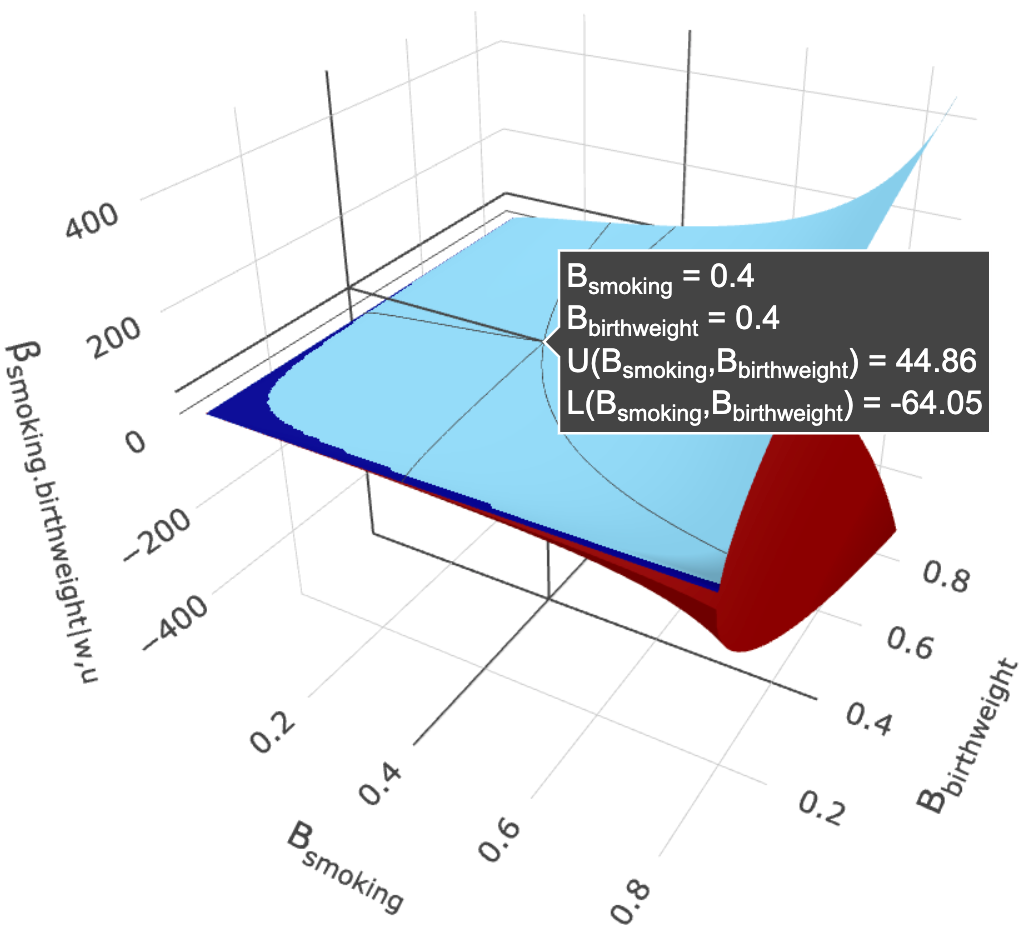}
    \caption{A 3D plot showing $L(B_{\textrm{smoking}}, B_\textrm{birthweight})$ in red, $U(B_{\textrm{smoking}}, B_\textrm{birthweight})$ in blue, and the region where $U(B_{\textrm{smoking}}, B_\textrm{birthweight})< 0$ in a darker shade of blue.}
    \label{3dplot1}
\end{figure}
\begin{figure}[ht]
    \centering
    \includegraphics[width=.6\linewidth]{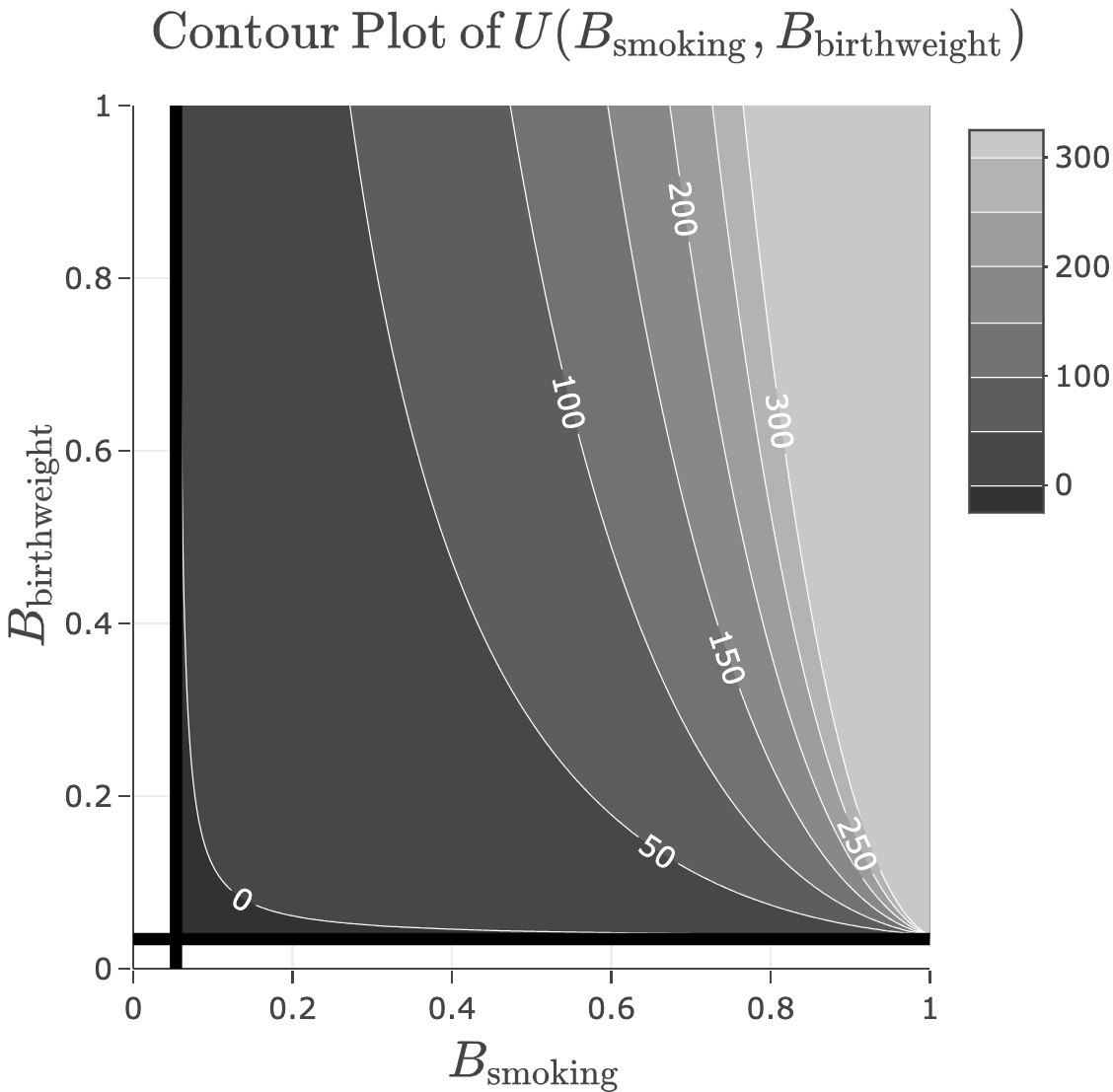}
    \caption{A contour plot displaying contour lines of $U(B_{\textrm{smoking}} , B_\textrm{birthweight})$.}
    \label{contour1}
\end{figure}

While the 3-d plot in Figure \ref{3dplot1} is made from information that is sufficient for a complete and thorough sensitivity analysis, subject matter experts may prefer to consider plausible and feasible values of $(B_{\textrm{smoking}},B_{\textrm{birthweight}})$, one pair at a time. Suppose for instance that subject matter experts specify $(B_{\textrm{smoking}},B_{\textrm{birthweight}})=(0.4,0.4)$. In that case the overall algorithm produces (in units of g/(cig/day))
\begin{equation}\label{app1result}L(0.4,0.4)=-64.05\le\beta_{\textrm{smoking}.\textrm{birthweight}|w,u}\le44.86=U(0.4,0.4).\end{equation}
We see in (\ref{app1result}) that the known randomness is insufficient to identify the sign or direction of $\beta_{\textrm{smoking}.\textrm{birthweight}|w,u}$. By insufficient randomness we mean that with bounds $B_{\textrm{smoking}}=0.4$ and $B_{\textrm{birthweight}}=0.4$ and coefficients of determination $(R^2_{w;\textrm{smoking}},R^2_{w;\textrm{birthweight}})=(0.05,0.03)$ as computed from $[y=\textrm{birthweight},x=\textrm{smoking},w]$, the coefficients of partial determination in (\ref{mainpropeq}) could be relatively high, as seen in the following: 
\begin{equation}\label{highy}R^2_{r(u,w);r(\textrm{birthweight};w)}=\frac{R^2_{w,u;\textrm{birthweight}}-R^2_{w;\textrm{birthweight}}}{1-R^2_{w;\textrm{birthweight}}}\le\frac{B_{\textrm{birthweight}}-R^2_{w;\textrm{birthweight}}}{1-R^2_{w;\textrm{birthweight}}}=0.38\end{equation}
and 
\begin{equation}\label{highx}R^2_{r(u,w);r(\textrm{smoking};w)}=\frac{R^2_{w,u;\textrm{smoking}}-R^2_{w;\textrm{smoking}}}{1-R^2_{w;\textrm{smoking}}}\le\frac{B_{\textrm{smoking}}-R^2_{w;\textrm{smoking}}}{1-R^2_{w;\textrm{smoking}}}=0.37.\end{equation}
With the possibility of that much residual and unaccounted for determinism (c.f. \citep{PearlDM}, \citep{Knaeble2023}) comes enough space to construct theoretical $u$ vectors to make the interval of (\ref{app1result}) fairly wide. Due to the large sample size we have a narrow confidence interval but a very wide confounding interval in this case. Keep in mind, however, that the confounding interval in (\ref{app1result}) depends on the specified bounds $(B_{\textrm{smoking}},B_{\textrm{birthweight}})=(0.4,0.4)$.

Instead of specifying values for $B_{\textrm{smoking}}$ and $B_{\textrm{birthweight}}$ a subject matter specialist can specify a threshold $\beta^\star_{\textrm{smoking}.\textrm{birthweight}}$ of practical significance for $\beta_{\textrm{smoking}.\textrm{birthweight}|w,u}$. Suppose for instance that a causal effect below $\beta^\star_{\textrm{smoking}.\textrm{birthweight}}=-1$ g/(cig/day) is practically significant. The information behind Figure \ref{3dplot1} and Figure \ref{contour1} then produces the contour of feasible $(B_{\textrm{smoking}},B_{\textrm{birthweight}})$ values satisfying \begin{equation}\label{contoura}U(B_{\textrm{smoking}},B_{\textrm{birthweight}})=\beta^\star_{\textrm{smoking}.\textrm{birthweight}}=-1 \textrm{~g/(cig/day)}.\end{equation} 
That contour bounds the gray region of $(R^2_{w,u;\textrm{smoking}},R^2_{w,u;\textrm{birthweight}})$ values pictured in Figure \ref{region1}. 
We may ask whether there exists a set of unmeasured covariates $u$ such that $\{w,u\}$ is admissible and $(R^2_{w,u;\textrm{smoking}},R^2_{w,u;\textrm{birthweight}})$ is within the gray region of Figure \ref{region1}. Because the gray region of Figure \ref{region1} is small and confined to where coefficients of determination are smaller, it seems implausible that we could find such a $u$ where $(R^2_{w,u;\textrm{smoking}},R^2_{w,u;\textrm{birthweight}})$ is in the gray region and $\{w,u\}$ is admissible, which means that the methodology is unlikely to demonstrate that smoking during pregnancy causes low birth weight. What blocks causal inference here is that the residual correlation of Figure \ref{residuals1} is relatively weak, while for any $u$ that makes $\{w,u\}$ admissible the coefficients of partial determination $R^2_{r(u,w);r(\textrm{smoking};w)}$ and $R^2_{r(u,w);r(\textrm{birthweight};w)}$ are plausibly of larger magnitude. For causal inference we want the opposite. What facilitates causal inference is a $w$ that produces strong residual correlation between $r(x;w)$ and $r(y;w)$ and then a $u$ where $\{w,u\}$ is admissible and the coefficients of partial determination $R^2_{r(u,w);r(x;w)}$ and $R^2_{r(u,w);r(y;w)}$ are smaller. That situation can arise in natural experiments, as demonstrated in the next example application. 
\begin{figure}[ht]
    \centering
    \includegraphics[width=.5\linewidth]{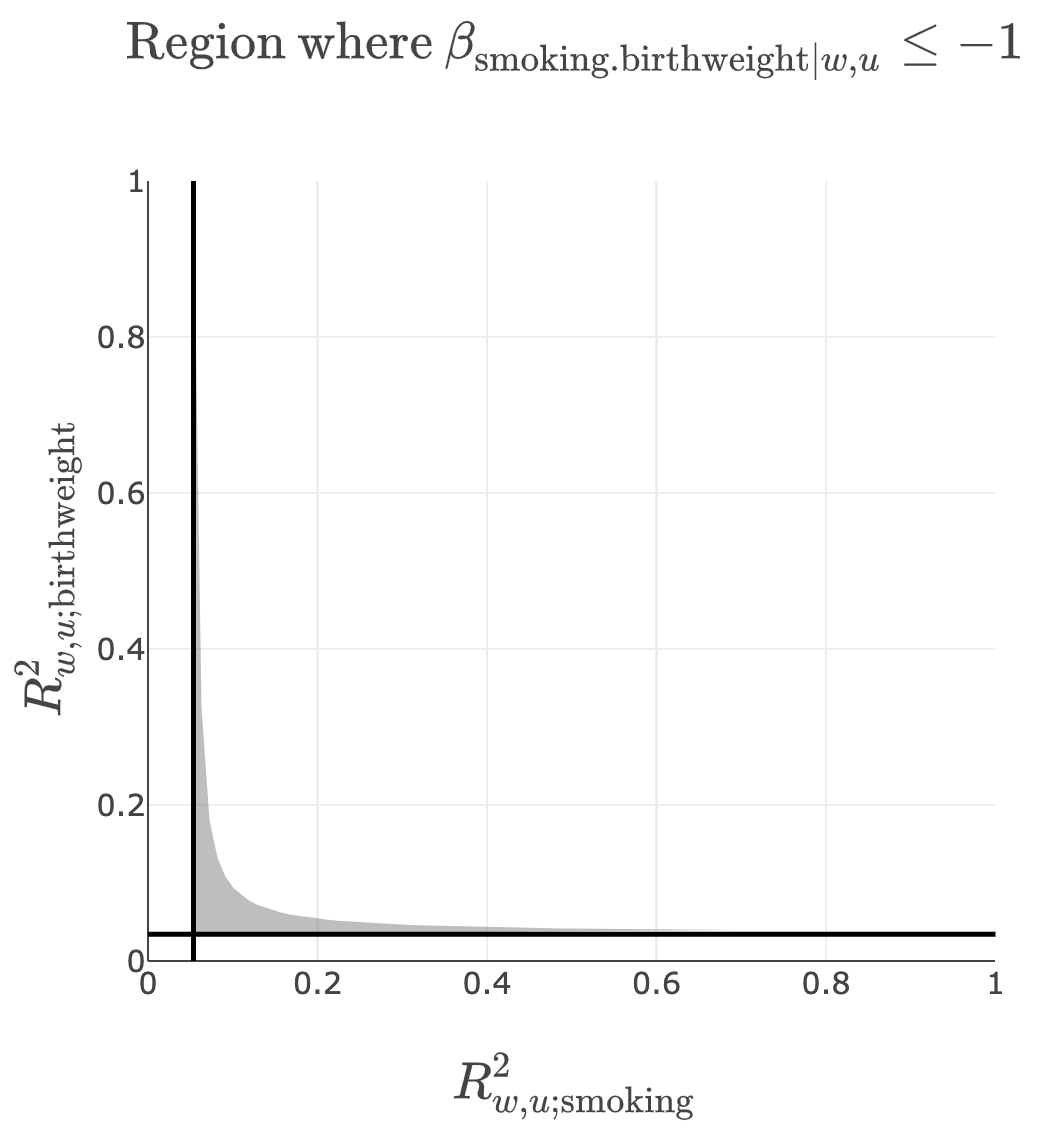}
    \caption{The (small) shaded region consists of $(R^2_{w,u;\textrm{smoking}},R^2_{w,u;\textrm{birthweight}})$ values that guarantee the practical significance of $\beta_{\textrm{smoking}.\textrm{birthweight}|w,u}$.}
    \label{region1}
\end{figure}
\newpage
\subsection{Wind and air quality} 
\label{wind}
Here we apply our methodology in an example context that involves natural randomness of meteorological processes. Our goal is to clarify why natural, natural experiments \citep{Rosen} lead to compelling causal inferences from observational data. The specific, causal question that we are asking is the following: what is the causal effect of wind on air quality in the Salt Lake City metropolitan area of Utah, USA? In the winter months of December and January the residents of Salt Lake City are exposed to winter inversions that trap air pollutants within an urban area at the base of the Wasatch mountains; See Figure \ref{Inv}. When an inversion is ongoing it can sometimes be difficult to forecast which incoming cold fronts will be strong enough to blow away the pollutants and improve the air quality.
\begin{figure}[ht]
\centering
\includegraphics[width=.6\linewidth]{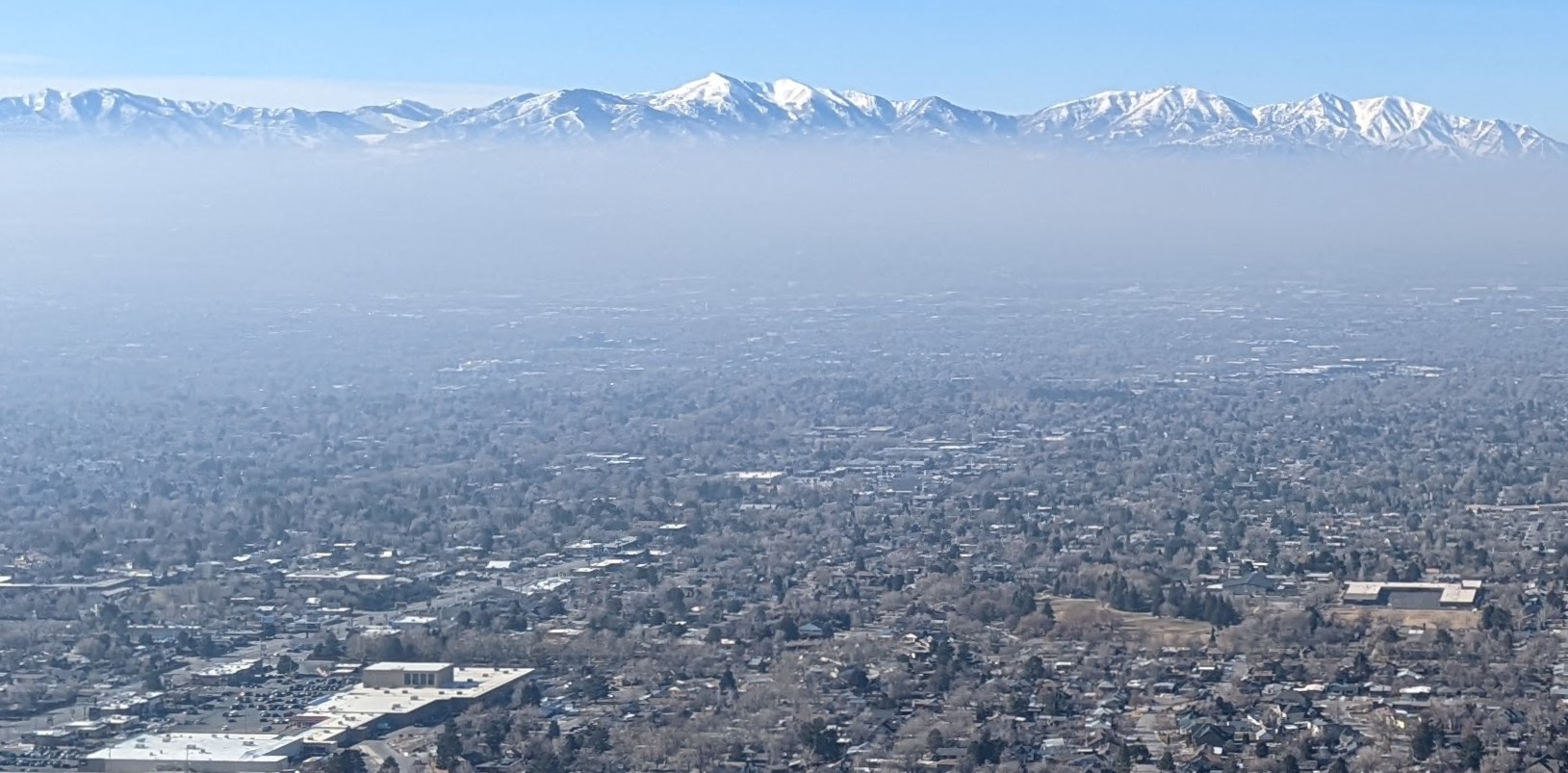}
\caption{A picture taken in January of 2022 showing a typical winter inversion causing poor air quality in Salt Lake City, Utah (USA) \citep{Williams2022}.}
\label{Inv}
\end{figure}

It is understood that the air pollution in Salt Lake City is caused by human activity, and much data has been collected. We have accessed data from the Hawthorn station in Salt Lake City \cite{EPA_AQS}. We are restricting our attention to measurements made during the winter months of December or January between the years of 2013 and 2024, inclusive. The treatment variable $x$ is average wind speed on Wednesdays, denoted with ``wind'' and measured in knots (kts). A common measure of the severity of particulate air pollution is known as PM2.5 and it is measured in micrograms per cubic meter (mcg/m3). We have obtained only those PM2.5 measurements recorded just after the conclusion of each Wednesday during the first hour after midnight of each Thursday. Our outcome variable $y$ is the natural logarithm of those PM2.5 values, denoted with ``$\ln(\textrm{PM2.5})$''. We introduce a unit for the $\ln(\textrm{PM2.5})$ values called the log concentration (lc) so that for any measured value $\textrm{PM2.5}>0$ we have $\ln(\textrm{PM2.5})\textrm{lc}=\textrm{PM2.5mcg/m3}.$ The sole covariate in $w$ is the minimum over each Tuesday's measurements of air pressure measured in millibars and denoted with ``pressure''. Over a range of typical wind values the correlation between wind and $\ln(\textrm{PM2.5})$ is negative, but PM2.5 can spike on extremely windy days due to blowing dust, and for that reason we excluded $7$ outliers defined statistically as those wind values exceeding the threshold of $Q3+1.5\times IQR$, where $Q3$ is the third quartile of wind and $IQR$ is the interquartile range of wind. After removing those outliers there were $n=77$ weeks of data to analyze. We intentionally accessed data at (regular) weekly time intervals to obtain approximately independent observations of each variable over time to better justify our use of the principle of least-squares. As done in the previous section, we ignore sampling error and focus instead on error due to unmeasured attributes.

We regress $\ln(\textrm{PM2.5})$ onto wind and pressure to obtain
$\beta_{\textrm{wind}.\ln(\textrm{PM2.5})|\textrm{pressure}}=-0.78$ lc/kt. By Proposition \ref{Danny} we have  $\beta_{\textrm{wind}.\ln(\textrm{PM2.5})|\textrm{pressure}}=\beta_{r(\textrm{wind};\textrm{pressure}).r(\ln(\textrm{PM2.5});\textrm{pressure})}=-0.78$ lc/kt. That adjusted slope is visible in Figure \ref{residuals2}. 
\begin{figure}[ht]
    \centering
    \includegraphics[width=.6\linewidth]{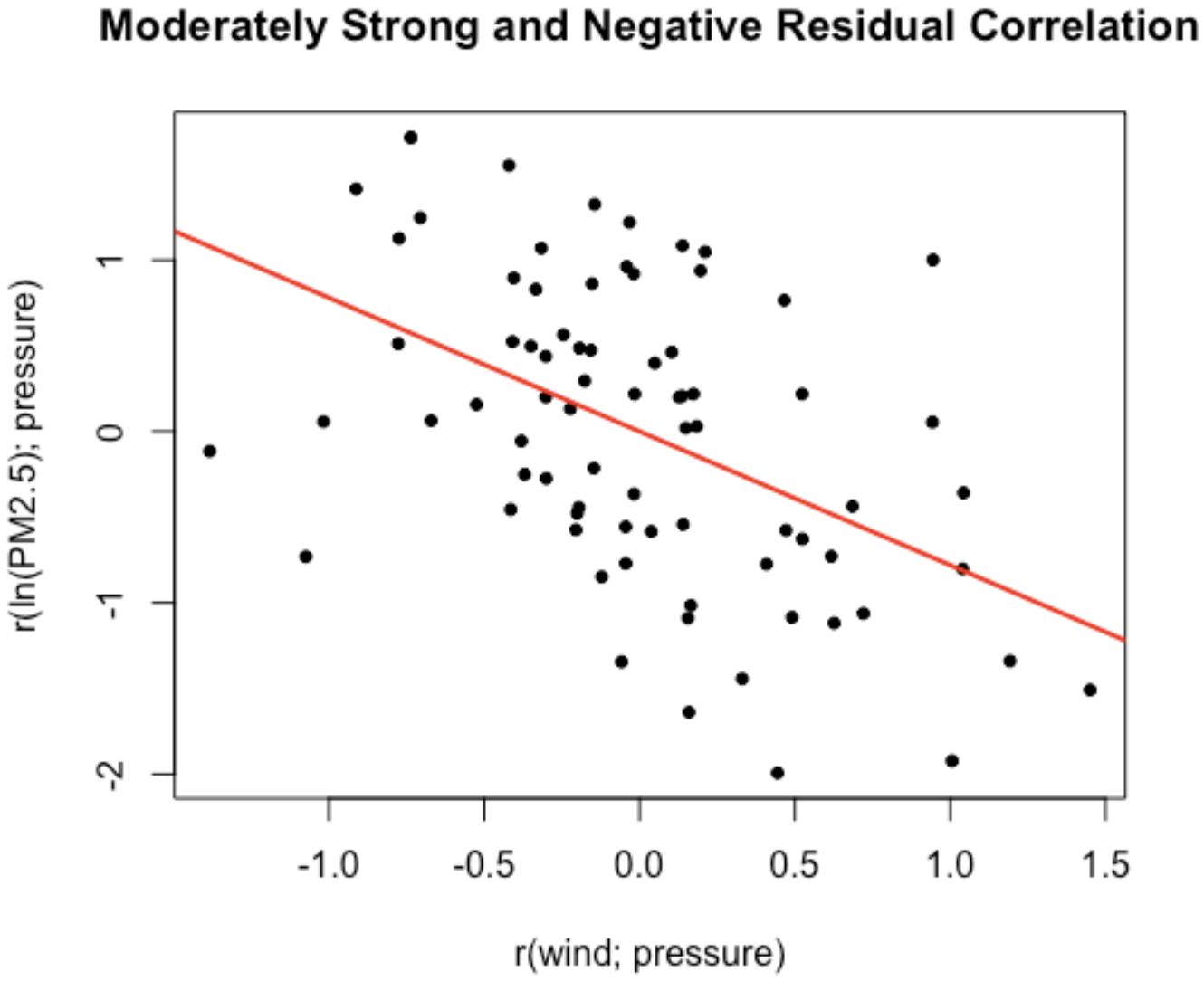}
    \caption{A scatter plot of $r(\textrm{wind};\textrm{pressure})$ vs $r(\ln(\textrm{PM2.5}); \textrm{pressure})$ along with the least-squares regression line with slope $\beta_{\textrm{wind}.\ln(\textrm{PM2.5})|\textrm{pressure}}=\beta_{r(\textrm{wind};\textrm{pressure}).r(\ln(\textrm{PM2.5});\textrm{pressure})}=-0.78$ lc/kt.}
    \label{residuals2}
\end{figure}
Conditional on pressure there is a negative relationship between wind and $\ln(\textrm{PM2.5})$. Our causal question is whether $$\beta_{r(\textrm{wind};\textrm{pressure}).r(\ln(\textrm{PM2.5});\textrm{pressure})}=-0.78\textrm{~lc/kt}$$ admits a causal interpretation. 
From the residuals pictured in Figure \ref{residuals1} we compute the ratio of standard deviations $\sigma_{r(\ln(\textrm{PM2.5});\textrm{pressure})}/\sigma_{r(\textrm{wind};\textrm{pressure})}= 1.62$ and the correlation coefficient $$\rho(r(\textrm{wind};\textrm{pressure}),r(\ln(\textrm{PM2.5});\textrm{pressure}))=-0.48$$ and then fix those values throughout the remainder of the analysis in this Subsection \ref{wind}. We have computed $R^2_{\textrm{pressure};\textrm{wind}}=0.14$ and $R^2_{\textrm{pressure};\ln(\textrm{PM2.5})}=0.28$ from the measured data. For each feasible (see (\ref{upperbounds})) pair $(B_{\textrm{wind}},B_{\ln(\textrm{PM2.5})})$ satisfying
\[R^2_{\textrm{pressure};\textrm{wind}}=0.14\leq B_{\textrm{wind}}<1\textrm{~and~}R^2_{\textrm{pressure};\ln(\textrm{PM2.5})}=0.28\leq B_{\ln(\textrm{PM2.5})} <1\] we utilize the overall algorithm of Subsection \ref{summary} to compute \[L(B_{\textrm{wind}}, B_{\ln(\textrm{PM2.5})})\textrm{~and~}U(B_{\textrm{wind}}, B_{\ln(\textrm{PM2.5})})\] of (\ref{partialinterval}) and plot the results in Figure \ref{3dplot2}. The contours of $U(B_{\textrm{wind}}, B_{\ln(\textrm{PM2.5})})$ are shown in Figure \ref{contour2}. 
\begin{figure}[b]
    \centering
    \includegraphics[width=.65\linewidth]{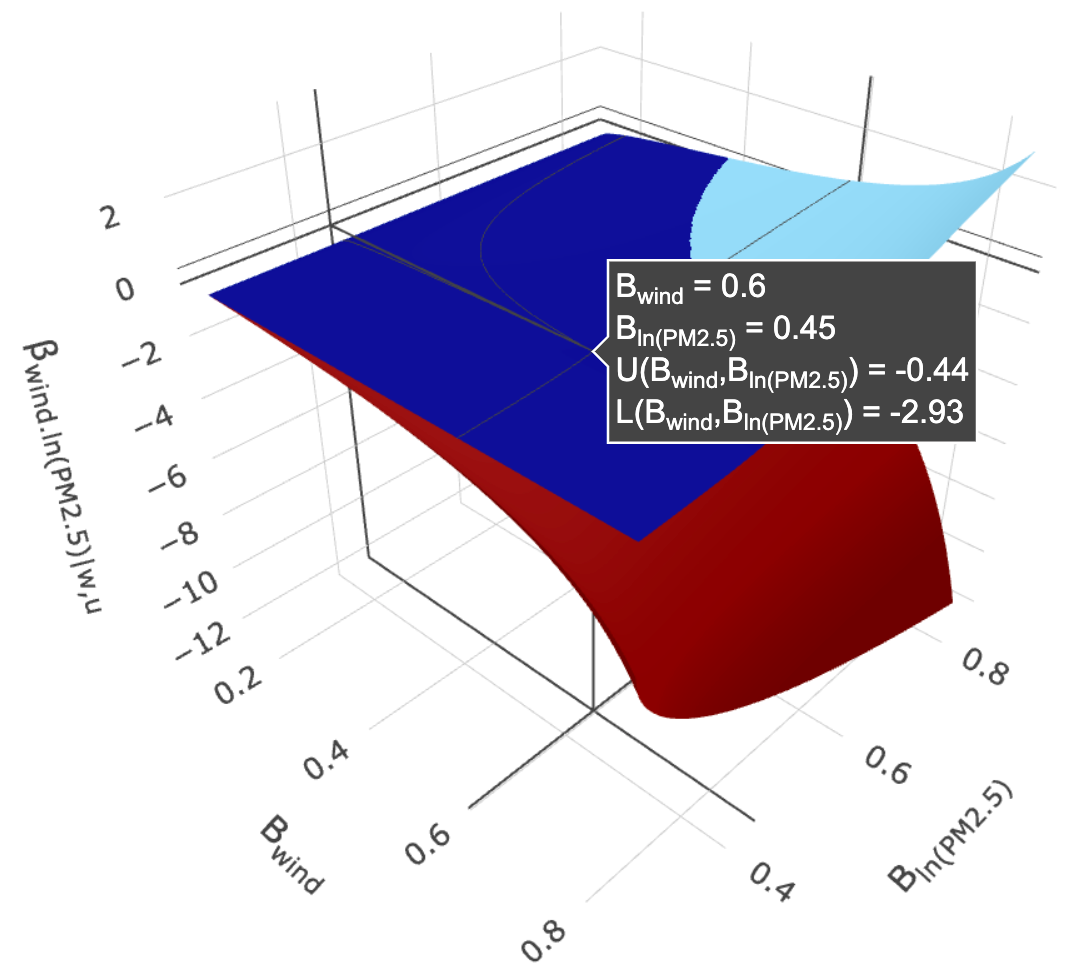}
    \caption{A 3D plot showing $L(B_{\textrm{wind}}, B_{\ln(\textrm{PM2.5})})$ in red, $U(B_{\textrm{wind}}, B_{\ln(\textrm{PM2.5})})$ in blue, and the region where $U(B_{\textrm{wind}}, B_{\ln(\textrm{PM2.5})})<0$ in a darker shade of blue.}
    \label{3dplot2}
\end{figure}
\begin{figure}[t]
    \centering
    \includegraphics[width=.6\linewidth]{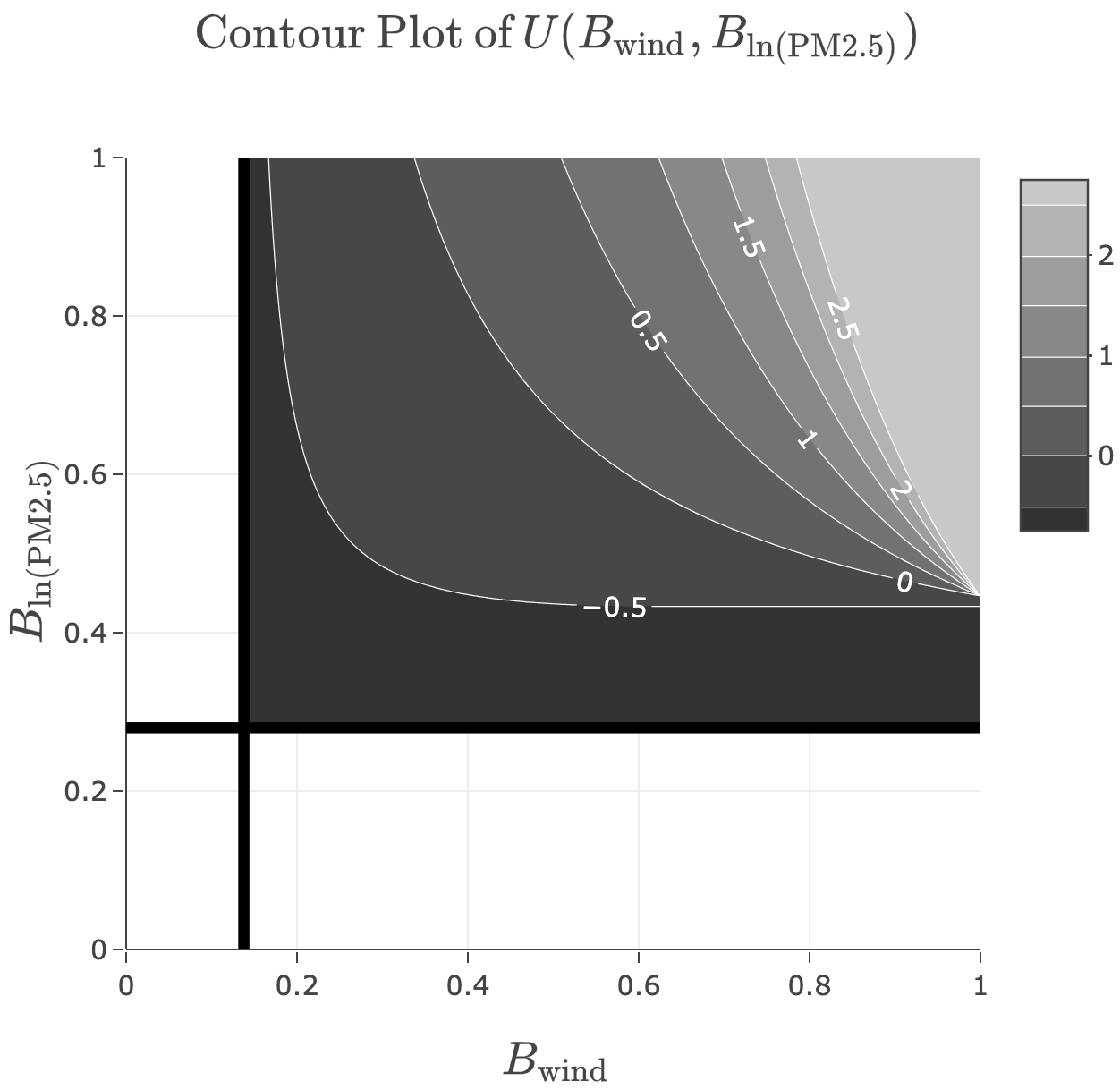}
    \caption{A contour plot displaying contour lines of $U(B_{\textrm{wind}} , B_{\ln(\textrm{PM2.5})})$.}
    \label{contour2}
\end{figure}

While the 3-d plot in Figure \ref{3dplot2} is made from information that is sufficient for a complete and thorough sensitivity analysis, subject matter experts may prefer to consider plausible and feasible pairs of $(B_{\textrm{wind}},B_{\ln(\textrm{PM2.5})})$ one at a time. Suppose for instance that subject matter experts specify $(B_{\textrm{wind}},B_{\ln(\textrm{PM2.5})})=(0.60,0.45)$. In that case the overall algorithm produces (in lc/kt units)
\begin{equation}\label{app1result2}L(0.60,0.45)=-2.93\le\beta_{\textrm{wind}.\ln(\textrm{PM2.5})|\textrm{pressure},u} \le -0.44=U(0.60,0.45).\end{equation}
From (\ref{app1result2}) we can identify the sign or direction of $\beta_{\textrm{wind}.\ln(\textrm{PM2.5})|\textrm{pressure},u}$, for any $u$ satisfying $R^2_{\textrm{pressure},u;\textrm{wind}}\leq 0.60$ and $R^2_{\textrm{pressure},u;\ln(\textrm{PM2.5})}\leq 0.45$. 

It could be argued, since meteorological processes are more random, that the bounds can be lowered to $(B_{\textrm{wind}},B_{\ln(\textrm{PM2.5})})=(0.25,0.40)$. Then the overall algorithm produces (in lc/kt units)
\begin{equation}\label{app1result3}
L(0.25,0.40)=
-1.17\le\beta_{\textrm{wind}.\ln(\textrm{PM2.5})|\textrm{pressure},u}\le -0.62
=
U(0.25,0.40).\end{equation} In (\ref{app1result3}) we see much more precise partial identification made possible by lower $(B_{\textrm{wind}},B_{\ln(\textrm{PM2.5})})$ values. The possibility of lower $(B_{\textrm{wind}},B_{\ln(\textrm{PM2.5})})$ values arises due to the natural randomness in meteorological processes. 

Instead of specifying values for $B_{\textrm{wind}}$ and $B_{\ln(\textrm{PM2.5})}$ a subject matter specialist can specify a threshold $\beta^\star_{\textrm{wind}.\ln(\textrm{PM2.5})}$ of practical significance for $\beta_{\textrm{wind}.\ln(\textrm{PM2.5})|\textrm{pressure},u}$. Suppose for instance that a causal effect below $\beta^\star_{\textrm{wind}.\ln(\textrm{PM2.5})}=-0.1$ lc/kt is practically significant. The information behind Figure \ref{3dplot2} and Figure \ref{contour2} then produces the contour of feasible $(B_{\textrm{wind}},B_{\ln(\textrm{PM2.5})})$ values satisfying \begin{equation}\label{contourb}U(B_{\textrm{wind}},B_{\ln(\textrm{PM2.5})})=\beta^\star_{\textrm{wind}.\ln(\textrm{PM2.5})}=-0.1 \textrm{~lc/kt}.\end{equation} 
That contour bounds the gray region of $(R^2_{\textrm{pressure},u;\textrm{wind}},R^2_{\textrm{pressure},u;\ln(\textrm{PM2.5})})$ values pictured in Figure \ref{region2}. The gray region of Figure \ref{region2} is relatively large, reflecting the moderately strong correlation in Figure \ref{residuals2}. 
Also, due to the natural randomness of meteorological phenomena there could be a $u$ that makes $\{w,u\}$ admissible and places $R^2_{\textrm{pressure},u;\textrm{wind}}$ and $R^2_{\textrm{pressure},u;\ln(\textrm{PM2.5})}$ within the gray region of Figure \ref{region2}. There is thus relatively strong evidence to support inference of a practically significant causal effect of wind on  $\ln(\textrm{PM2.5})$.
\begin{figure}[ht]
    \centering
    \includegraphics[width=.5\linewidth]{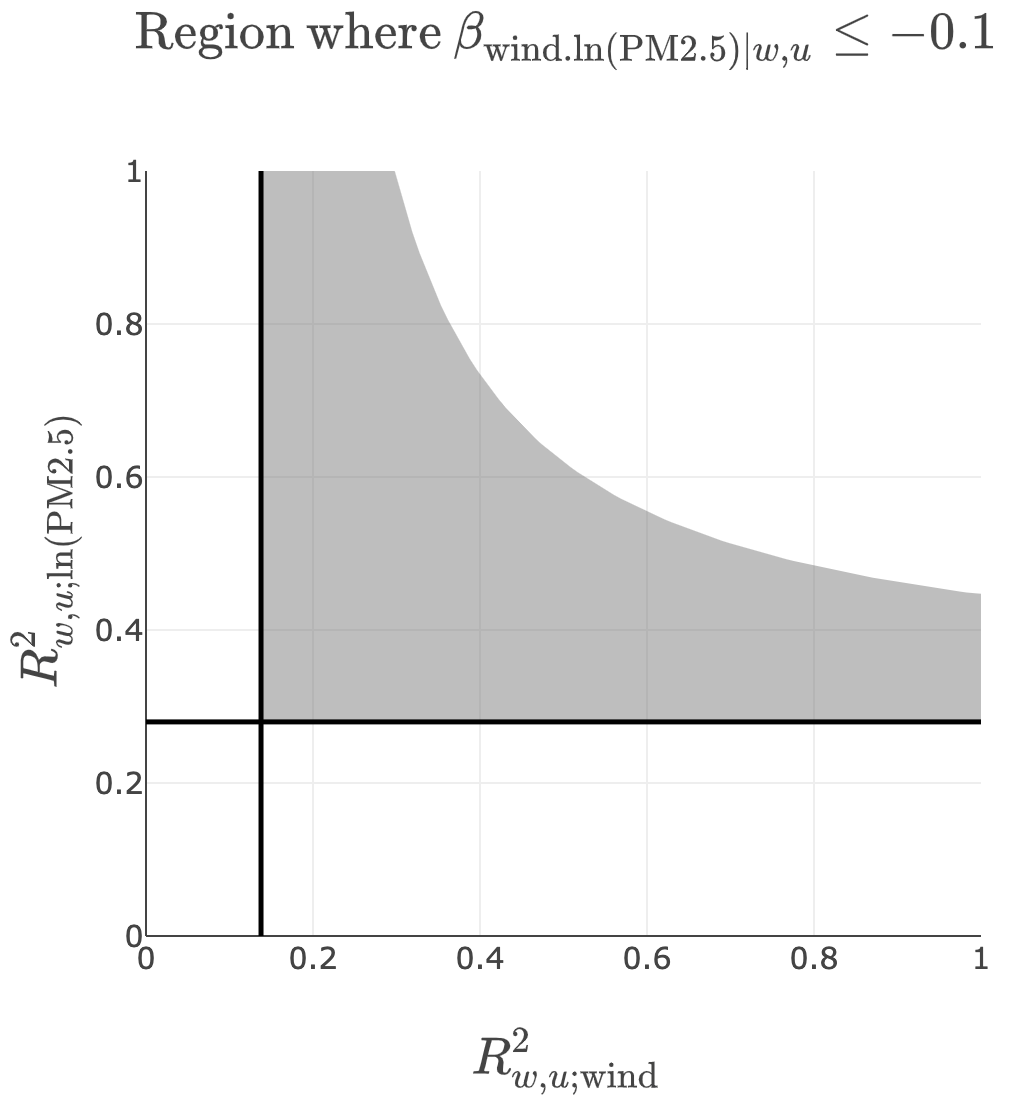}
    \caption{The (large) shaded region consists of $(R^2_{\textrm{pressure},u;\textrm{wind}},R^2_{\textrm{pressure},u;\ln(\textrm{PM2.5})})$ values that guarantee the practical significance of $\beta_{\textrm{wind}.\ln(\textrm{PM2.5})|\textrm{pressure},u}$}
    \label{region2}
\end{figure}

\newpage
\section{Discussion}
\label{discussion}
In this paper we have shown how the coefficients of determination $R^2_{w,u;x}$ and $R^2_{w,u;y}$ can be used as sensitivity parameters within a multiple regression context to analyze omitted variables bias. Uncertainty in causal interpretation of the slope coefficient $\beta_{x.y|w}$ (see (\ref{pramod})), due to unmeasured confounding by a set of confounding variables $u$, can be assessed with the methodology described in Section \ref{methods}. The example applications of Section \ref{application} demonstrate how our methodology can partially identify the slope coefficient $\beta_{x.y|w,u}$ provided there are reliable upper bounds $B_x$ and $B_y$ (see \ref{upperbounds}) on $R^2_{w,u;x}$ and $R^2_{w,u;y}$. Those same applications demonstrate also how to conduct sensitivity analysis with variable bounds $B_x$ and $B_y$. We have assumed the existence of an unmeasured set of covariates $u$ so that the covariate set $\{w,u\}$ is admissible. We have defined the covariates in $w$ and $u$ for each individual at times prior to treatment, preferably prior to a naturally random process that gives rise to the treatment. Awareness of that randomness leads to upper bounds $B_x$ and $B_y$ that are closer to $R^2_{w,u;x}$ and $R^2_{w,u;y}$ and consequently, via our methodology, less uncertainty in causal interpretation of $\beta_{x.y|w}$. With perfect bounds $R^2_{w,u;x}=B_x$ and $R^2_{w,u;y}=B_y$ there is no room for residual determinism, only residual randomness, and our partial identification methodology becomes a methodology to point identify $\beta_{x.y|w,u}$ with $\beta_{x.y|w}$.

It may not be appropriate, however, to accumulate randomness by pushing back the times of definition for the covariates in $\{w,u\}$ way back to long before the individual treatment times. Doing so may increase the likelihood of a Stable Unit Treatment Value Assumption (SUTVA) violation; see \citep{Imbens2015}. Something like that must be done though if we are to harness the randomness of a natural, natural experiment. Our partial identification of $\beta_{x.y|w,u}$ with the methodology of this paper relies on SUTVA in the following way. We explain under the assumption that $B_x$ is a tight upper bound for $R^2_{w,u;x}$. The explanation with respect to $y$ is similar. Since $B_x-R^2_{w,u;x}$ is positive but nearly zero, as long as $1-B_x$ is not too small then $\frac{1-B_x}{1-R^2_{w,u;x}}$ is large (near one) meaning that nearly $100\%$ of the residual variance in $x$ can not be predicted from $(w,u)$ at their times of definition. That means that residual variance in $x$ is largely due to chance events that occur after the times of definition of $(w,u)$. Those events may be instruments or prods in the sense of \cite{Pim}, and some sort of exclusion restriction \citep{Imbens2015} is required to warrant causal inference. The possibly random events that arise after the definition times of $(w,u)$ for each individual could themselves be confounding the observed correlation between residual $x$ and residual $y$ conditional on $(w,u)$. One way to address such a concern is to consider the potential outcome $y_i(x)$, meaning the $y$ value for individual $i$ that would occur if treatment for individual $i$ were contrary to fact set to $x$. But there is a subtlety relating the times of definition for those counterfactual potential outcomes. That subtlety is hidden over the short time frames associated with coin flips in a random experiment, but that subtlety becomes apparent over the longer time frames of naturally random processes like the meteorological processes of Section \ref{wind}. If each individual's potential outcomes are stable from the times of definition of $(w,u)$ until the moment $x$ is measured, inclusive of that moment, then our partial identification of $\beta_{x.y|w,u}$ would be more likely to admit a causal interpretation. This is the sense in which we are relying on SUTVA. 

Many methodologies for causal inference rely on SUTVA \citep{Imbens2015}, as we rely on SUTVA here, but we warn against causal interpretation of $\beta_{x.y|w,u}$ without first considering the validity of SUTVA. Likewise, we assume the existence of an unmeasured set of covariates $u$ so that $\{w,u\}$ is admissible in the sense of \citep{Pearl2009}, and that assumption may not be valid. Perhaps our most limiting assumption with regards to causal inference is an implicit assumption of homogeneous causal effects made by our choice of regression models. Regression based techniques are known to be less reliable for causal inference \citep{Rose10}, and there are methodologies similar to what we have described here that allow for heterogeneous effects and do not utilize regression models \citep{Knaeble2023, Knaeble2024, 10.48550/arxiv.2407.19057}. Nevertheless, there is an elegance to least-squares \citep{KD, Knaeble2020}, and regression is still useful \citep{Angrist2009}. To limit our reliance on SUTVA we have presented our methodology as a technique for partial identification of $\beta_{x.y|w,u}$ without requiring a causal interpretation; that is our methodology may be used in a regression context, assuming the principle of least-squares, to assess what might happen to $\beta_{x.y|w}$ if the additional covariates of $u$ are added to the model, and that technical assessment is valid regardless of how $\beta_{x.y|w}$ and $\beta_{x.y|w,u}$ are interpreted. In other words, one may interpret our algorithm as an interesting mathematical result regarding multiple least-squares projections of various vectors onto various hyperplanes. The basic CI algorithm described in Subsection \ref{CI} and first published in \citep{Knaeble2020} supports, for instance, efficient search for maximum and minimum slope coefficients over a large space of models \citep{KBB}. In \citep{KBB} the space of models is too large for a brute force search over each possible model that could be fit to the measured data, while here our methodology is concerned not so much with measured data but primarily with understanding the sensitivity of conclusions to unmeasured confounding.

There are a variety of methods for sensitivity analysis of unmeasured confounding in a regression context. We restate some results here in our notation. Economists are familiar \citep{Angrist2009} with the following formula for omitted variable bias: \begin{equation}\label{simple}\beta_{x.y}-\beta_{x.y|u_1}=\beta_{u_1.y|x}\beta_{x.u_1}.\end{equation} Statisticians have derived \citep{HHH} the formula \begin{equation}\label{Heq}\hat{\beta}_{x.y|w}-\hat{\beta}_{x.y|w,u_1}=\textrm{SE}(\hat{\beta}_{x.y|w})t(u_1)\rho(r(y;x,w),r(u_1;x,w)),\end{equation} where the population is finite, SE stands for the standard error, and $t(u_1)$ is student's t statistic associated with $u_1$ when $x$ is regressed onto $[w,u_1]$. That formula generalizes to allow multivariate $u$ in place of $u_1$ as described in \citep[Section 4]{HHH}. Another economics paper \citep{Oster2017} introduces the formula \begin{equation}\label{Osterresult}\beta_{x.y|w,u_1}=\beta_{x.y|w}-(\beta_{x.y}-\beta_{x.y|w})\frac{R^2_{x,w,u_1;y}-R^2_{x,w;y}}{R^2_{x,w;y}-R^2_{x;y}},\end{equation} but it relies on some assumptions \citep[Section 3.2]{Oster2017}. There is another finite-population equation \begin{equation}\label{Cin}(\hat{\beta}_{x.y|w}-\hat{\beta}_{x.y|w,u_1})^2=\frac{\rho(r(y;x,w),r(u_1;x,w))^2\rho(r(x;w),r(u_1,w))^2}{1-\rho(r(x;w),r(u_1,w))^2}\textrm{SE}(\hat{\beta}_{x.y|w})^2\textrm{df}\end{equation} published in \citep{CH}, where SE is again a standard error and df is the degrees of freedom when $y$ is regressed onto $[x,w]$, and generalization from $u_1$ to multivariate $u$ is described in \citep[Section 4.5]{CH}. Social scientists are aware \citep{Frank} of the following formula for an adjusted slope coefficient: \begin{equation}\label{Franki}\beta_{x.y|u}=\left(\frac{\sigma_y}{\sigma_x}\right)\frac{\rho(x,y)-\rho(\hat{x}_u,x)\rho(\hat{x}_u,y)}{1-\rho^2(\hat{x}_u,x)}.\end{equation} 
The formula in (\ref{Franki}) is similar to a formula published in \citep{KD}, which forms the basis for the CI algorithm described in Subsection \ref{CI} and first published in \citep{KOA}. Using Proposition \ref{Danny}, the formula of \citep{KD} can be expressed in the notation of this paper as 
\begin{align}
\label{form}
&\beta_{x.y|w,u}=\beta_{r(x;w).r(y;w)|r(u;w)}=\\
&\nonumber\left(\frac{\sigma_{r(y;w)}}{\sigma_{r(x;w)}}\right)\frac{\rho(r(x;w),r(y;w))-R_{r(u;w);r(x;w)}R_{r(u;w);r(y;w)}\rho\left(\hat{r(x;w)}_{r(u;w)},\hat{r(y;w)}_{r(u;w)}\right)}{\left(1-R_{r(u;w);r(x;w)}^2\right)}.
\end{align}

The CI algorithm described in Subsection \ref{CI} works by using (\ref{form}) to optimize $\beta_{r(x;w).r(y;w)|r(u;w)}$ subject to the default bounds $-1\leq\rho\left(\hat{r(x;w)}_{r(u;w)},\hat{r(y;w)}_{r(u;w)}\right)\leq1$ and also user specified bounds on $R^2_{r(u;w);r(x;w)}$ and $R^2_{r(u;w);r(y;w)}$. Those bounds are obtained from (\ref{upperbounds}) using Proposition \ref{mainprop} and Corollary \ref{cor}. Note that $R^2_{r(u;w);r(x;w)}$ and $R^2_{r(u;w);r(y;w)}$ are coefficients of partial determination, as is $\rho(r(y;x,w),r(u_1;x,w))^2$ in (\ref{Cin}); see also (\ref{mainpropeq}). When subject matter knowledge is able to improve on the default bounds and more tightly bound $\rho\left(\hat{r(x;w)}_{r(u;w)},\hat{r(y;w)}_{r(u;w)}\right)$ then more precise partial identification is possible, and the available algorithms at \citep{Hughes25} are flexible enough to accept those tighter bounds as additional inputs for enhanced inference. 

Many of those formulas and equations in (\ref{simple}) through (\ref{Franki}) involve $\beta_{x.y|w}-\beta_{x.y|w,u}$ or a similar expression for omitted variables bias. Those formulas and expressions allow one to conclude that $\beta_{x.y|w,u}$ is close to $\beta_{x.y|w}$ from knowledge of the strengths of various linear relationships between variables. The idea behind our partial identification of $\beta_{x.y|w,u}$ is similar, but we do not bound $R^2_{w,u;x}$ and $R^2_{w,u;y}$ from supposed knowledge of relationships between variables. Our methodology is applicable even if we don't know which variables are in $u$. We bound $R^2_{w,u;x}$ and $R^2_{w,u;y}$ via (\ref{upperbounds}) and (\ref{mainpropeq}) by reasoning about determinism and randomness in the process or processes generating $x$ and $y$. Our approach is similar to the reasoning in \citep{PearlDM}. 
If $R^2_{w;x}$ and $R^2_{w;y}$ are much smaller than what is reasonably possible for coefficients of determination then there may be sufficient determinism remaining to theoretically construct a $u$ to make $\beta_{x.y|w,u}$ far from $\beta_{x.y|w}$. However, if $R^2_{w;x}$ and $R^2_{w;y}$ are only slightly less than what is reasonably possible for coefficients of determination then there is likely insufficient remaining determinism to theoretically construct a $u$ to make $\beta_{x.y|w,u}$ far from $\beta_{x.y|w}$, or said differently there is sufficient residual randomness to conclude that $\beta_{x.y|w,u}$ is close to $\beta_{x.y|w}$. For further and more detailed reading about sufficient randomness for causal inference see \cite{Knaeble2023,Knaeble2024}.

Here is a related and practical way to implement our methodology based on the idea that partial identification becomes point identification when the coefficients of partial determination $R^2_{r(u,w);r(x;w)}$ and $R^2_{r(u,w);r(y;w)}$ approach zero; see (\ref{mainpropeq}). First search for natural sources of randomness \citep{Rosen}, and then use that randomness to bound $R^2_{w,u;x}$ and $R^2_{w,u;y}$ with upper bounds $B_x=R^2_{s;x}$ and $B_y=R^2_{s;y}$, where $s$ is a set of covariates containing $\{w,u\}$, i.e. $\{w,u\}\subseteq s$. Ideally, $s$ will not contain variables that are indicators of chance events nor any randomly generated variables \citep{Pim}, but rather $s$ should contain variables that are fixed attributes of individuals \citep[Section 4.3]{Knaeble2024}. The covariates in $s$ should be defined prior to the natural randomness, and the potential outcomes should remain stable over the time it takes for that natural randomness to play out. Then, consistent with the advice of \citep{rubin2009} and following the principles of \citep{VandPrinc}, we select as many confounding variables in $s$ as possible to be measured and included within $w$ in hopes that $R^2_{w;x}$ and $R^2_{w;y}$ increase to approach the upper bounds $B_x$ and $B_y$, so that any remaining, unmeasured covariates in a set $u\subseteq (s\setminus w)$ play less of a role. 
We thus make the coefficients of partial determination $R^2_{r(u,w);r(y;w)}$ and $R^2_{r(u,w);r(x;w)}$ approach zero to support more precise partial identification of $\beta_{x.y|w,u}$ with our overall algorithm.
\bibliographystyle{plain}


\end{document}